\documentclass[a4,12pt]{article}
\usepackage{graphicx,url,amsmath,amssymb,color,subfigure,anysize,amsthm}
\usepackage[affil-it]{authblk}
\marginsize{2cm}{2cm}{2cm}{2cm}

\newcommand{\eps}{\varepsilon}
\newcommand{\set}[1]{\left\{#1\right\}}

\newcommand{\p}{\partial}

\newcommand{\me}{\mathbf{e}}
\newcommand{\mf}{\mathbf{f}}
\newcommand{\mn}{\mathbf{n}}

\newcommand{\mt}{\mathbf{t}}
\newcommand{\mx}{\mathbf{x}}

\newcommand{\mz}{\mathbf{z}}

\newcommand{\mO}{\mathbf{O}}
\newcommand{\mU}{\mathbf{U}}
\newcommand{\mV}{\mathbf{V}}
\newcommand{\vt}{\boldsymbol{\theta}}
\newcommand{\vv}{\boldsymbol{\vartheta}}
\newcommand{\vp}{\boldsymbol{\varphi}}
\newcommand{\vx}{\boldsymbol{\xi}}

\newtheorem{thm}{Theorem}[section]

\newtheorem{lem}[thm]{Lemma}

\newtheorem{rem}[thm]{Remark}

\begin{document}

\title{Asymptotic properties of MUSIC-type imaging in two-dimensional inverse scattering from thin electromagnetic inclusions}

\author{Won-Kwang Park\thanks{e-mail: parkwk@kookmin.ac.kr}}
\affil{Department of Mathematics, Kookmin University, Seoul, 136-702, Korea.}

\date{}
\maketitle

%The abstract of your paper
\begin{abstract}
The main purpose of this paper is to study the structure of the well-known non-iterative MUltiple SIgnal Classification (MUSIC) algorithm for identifying the shape of extended electromagnetic inclusions of small thickness located in a two-dimensional homogeneous space. We construct a relationship between the MUSIC-type imaging functional for thin inclusions and the Bessel function of integer order of the first kind. Our construction is based on the structure of the left singular vectors of the collected multistatic response matrix whose elements are the measured far-field pattern and the asymptotic expansion formula in the presence of thin inclusions. Some numerical examples are shown to support the constructed MUSIC structure.
\end{abstract}

\section{Introduction}
One of the goals of the inverse scattering problem is to identify unknown properties (e.g., the shapes, material properties, locations, and constitutions) of electromagnetic targets from measured scattered field data. This problem is generally solved by Newton-type iteration schemes or level-set method involving minimization of the difference between the measured scattered data and the computed data by generating an admissible cost functional. Related works can be found in \cite{AGJKLY,BHL,B,CR,DL,GH2,K,M2,PL4,VXB,SZ} and references therein. To execute these schemes, the Fr{\'e}chet derivative must be evaluated at each iteration step, so the computational costs are large. Unfortunately, non-convergence or the appearance of several minima arises in the iteration procedure owing to the non-convex nature of the cost functional. Furthermore, \textit{a priori} information on unknown targets is essential to guarantee a successful reconstruction, and even when the above conditions are fulfilled, reconstruction will fail if the iteration process is begun with a bad initial guess. Hence, generating a good initial guess close to the expected conditions is a priority. For this purpose, alternative non-iterative imaging algorithms such as the linear sampling method, single- and multi-frequency based Kirchhoff and subspace migrations, and the topological derivative strategy have been developed.

In pioneering research \cite{D}, the MUltiple SIgnal Classification (MUSIC) algorithm has been investigated to find the locations of point-like scatterers. It was recently applied to various problems, for example, detection of antipersonnel mines buried in the ground \cite{AIL}, searching for the locations of small inclusions \cite{AILP,CZ,GH1,HM,SCC,ZC}, identifying internal corrosion in a pipeline \cite{AKKLV}, and shape reconstruction of arbitrarily shaped thin inclusions, cracks, and extended targets \cite{AGKPS,AKLP,HSZ,JP,PL1,PL3}. On the basis of these results, the locations of small inclusions can be accurately identified using MUSIC, but owing to the intrinsic resolution limit, the complete shape of extended targets cannot be imaged. Hence, the obtained results were based on good initial guesses, and iteration-based algorithms such as the level-set method were successfully executed (see
\cite{AGJKLY,BHL,PL4}).

Although the MUSIC algorithm offers good results for small and extended targets, a detailed structural analysis must be attempted because some phenomena cannot be explained using the traditional approach, for example, the unexpected appearance of artifacts or of two curves along the boundary of targets instead of the true shape (see \cite[Figure 9(b)]{HSZ}), or an image with poor resolution (see \cite[Section 4.4]{PL3}). Numerical results in existing works motivate us to explore some properties of the MUSIC-type algorithm for imaging the arbitrarily shaped thin penetrable electromagnetic inclusions and perfectly conducting cracks considered in \cite{PL1,PL3}. Our exploration is based on the rigorously derived asymptotic expansion formula in the presence of a thin inclusion \cite{BF} and physical factorization of the so-called multistatic response (MSR) matrix \cite{HSZ}. With this, we will establish a relationship between the MUSIC-type imaging functional and the Bessel function of integer order of the first kind, and identify its properties.

This paper is organized as follows. In section \ref{sec:2}, we introduce two-dimensional direct scattering problems, the asymptotic expansion formula in the presence of thin inclusions, and the MUSIC-type algorithm for imaging thin electromagnetic inclusions. In section \ref{sec:3}, we identify the structure of a MUSIC-type functional focused on imaging of thin penetrable electromagnetic inclusions by constructing a relationship with the Bessel function of integer order of the first kind, and discuss its properties. In section \ref{sec:4},  numerical simulation results are presented to support the identified structure. This paper ends with a short conclusion in section \ref{sec:5}.

\section{Direct scattering problems and MUSIC algorithm}\label{sec:2}
\subsection{Two-dimensional direct scattering problems and asymptotic expansion formula}
Suppose that an extended electromagnetic inclusion $\Gamma$ with a small (with respect to the given wavelength) thickness $2h$ is located in the two-dimensional homogeneous space $\mathbb{R}^2$. We assume that the shape of $\Gamma$ is characterized by the supporting smooth curve $\gamma$ such that (see FIG. \ref{ThinInclusion})
\[\Gamma=\{\mx+\eta\mn(\mx):\mx\in\gamma,~\eta\in(-h,h)\}.\]

\begin{figure}[!ht]
  \begin{center}
  \includegraphics[width=0.35\textwidth]{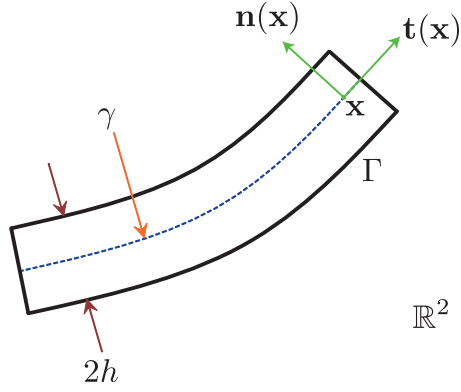}
  \caption{\label{ThinInclusion}Sketch of two-dimensional thin electromagnetic inclusion $\Gamma$.}
  \end{center}
\end{figure}

Throughout this paper, we assume that $\Gamma$ and $\mathbb{R}^2$ are classified by their dielectric permittivity and magnetic permeability at a given frequency $\omega=2\pi/\lambda$, where $\lambda$ denotes the wavelength. Let $0<\eps_0<+\infty$ and $0<\mu_0<+\infty$ denote the permittivity and permeability of $\mathbb{R}^2$, respectively; analogously, $0<\eps<+\infty$ and $0<\mu<+\infty$ represent those of $\Gamma$. Then, we can define the piecewise constant dielectric permittivity $\varepsilon(\mx)$ and magnetic permeability $\mu(\mx)$,
\begin{equation}\label{EPSMU}
\varepsilon(\mx)=\left\{\begin{array}{ccl}
\varepsilon_0&\mbox{for}&\mx\in\mathbb{R}^2\backslash\overline{\Gamma}\\
\varepsilon&\mbox{for}&\mx\in\Gamma
\end{array}\right.
\quad\mbox{and}\quad
\mu(\mx)=\left\{\begin{array}{ccl}
\mu_0&\mbox{for}&\mx\in\mathbb{R}^2\backslash\overline{\Gamma}\\
\mu&\mbox{for}&\mx\in\Gamma,
\end{array}\right.
\end{equation}
respectively. For convenience, we set $\eps_0=\mu_0=1$, $\eps>\eps_0$, and $\mu>\mu_0$.

For a given fixed frequency $\omega$ (we assume that the wave number $k=\omega\sqrt{\eps_0\mu_0}=\omega$), we let
\begin{equation}\label{IncidentField}
  u_0(\mx,\vt;\omega):=e^{i\omega\vt\cdot\mx},\quad\mx\in\mathbb{R}^2
\end{equation}
be a plane-wave incident field with the incident direction $\vt\in\mathbb{S}^1$, where $\mathbb{S}^1$ denotes the unit circle. Let $u(\mx,\vt;\omega)=u_0(\mx,\vt;\omega)+u_s(\mx,\vt;\omega)$ denote the time-harmonic total field that satisfies the following Helmholtz equation:
\begin{equation}\label{TotalField}
  \nabla\cdot\bigg(\frac{1}{\mu(\mx)}\nabla u(\mx,\vt;\omega)\bigg) +\omega^2\eps(\mx)u(\mx,\vt;\omega)=0,
\end{equation}
with transmission conditions on the boundary of $\Gamma$. Here, $u_s(\mx,\vt;\omega)$ denotes the unknown scattered field that satisfies the Sommerfeld radiation condition
\[\lim_{r\to0}\sqrt{r}\bigg(\frac{\p u_s(\mx,\vt;\omega)}{\p r}-i\omega u_s(\mx,\vt;\omega)\bigg)=0,\quad\vv=\frac{\mx}{|\mx|}\in\mathbb{S}^1\]
uniformly in all directions $\vv$. Notice that the above radiation condition implies the asymptotic behavior
\[u_s(\mx,\vt;\omega)=\frac{e^{i\omega|\mx|}}{\sqrt{|\mx|}}u_{\infty}(\vv,\vt;\omega)+o\left(\frac{1}{\sqrt{|\mx|}}\right)\quad\mbox{for all}\quad|\mx|\longrightarrow+\infty.\]
Then, according to \cite{BF}, the far-field pattern $u_{\infty}(\vv,\vt;\omega)$ can be written using the following asymptotic expansion formula:
\[u_{\infty}(\vv,\vt;\omega)=h\frac{\omega^2(1+i)}{4\sqrt{\omega\pi}} \int_\gamma\bigg((-\vv)\cdot\mathbb{M}(\mx)\cdot\vt+(\eps-1)\bigg)e^{i\omega(\vt-\vv)\cdot\mx}d\gamma(\mx)+o(h).\]
Here, the $2\times2$ symmetric matrix $\mathbb{M}(\mx)$ is defined as follows: for $\mx\in\gamma$, let $\mt(\mx)$ and $\mn(\mx)$ denote the unit tangent and normal vectors to $\gamma$ at $\mx$, respectively. Then
\begin{itemize}
  \item $\mathbb{M}(\mx)$ has eigenvectors $\mt(\mx)$ and $\mn(\mx)$.
  \item The eigenvalue corresponding to $\mt(\mx)$ is $2\left(\frac{1}{\mu}-1\right)$.
  \item The eigenvalue corresponding to $\mn(\mx)$ is $2\left(1-\mu\right)$.
\end{itemize}

\subsection{MUSIC-type imaging algorithm}
Next, we introduce the MUSIC algorithm for imaging $\Gamma$. For simplicity, suppose that we have $N$ incident and observation directions $\vt_l$ and $\vv_j$, respectively, for $j,l=1,2,\cdots,N$, and the incident and observation directions are the same, i.e., $\vv_j=-\vt_j$. In this paper, we consider the full-view inverse problem. Hence, we assume that $\set{\vt_n:n=1,2,\cdots,N}$ spans unit circle $\mathbb{S}^1$. Moreover, we assume that the supporting curve $\gamma$ is divided into $M$ different segments with sizes on the order of half the wavelength, $\lambda/2$. Then, keeping in mind the Rayleigh resolution limit from the far-field data, any detail less than one-half of the wavelength in size cannot be seen, and only one point at each segment is expected to contribute to the image space of the response matrix $\mathbb{K}$ (see \cite{AKLP,HSZ,PL1,PL3}, for instance). Each of these points, say $\mx_j$ for $j=1,2,\cdots,M$, can be imaged by MUSIC-type imaging. With this in mind, we consider the collected MSR matrix such that
\begin{equation}\label{MSR}
\mathbb{K}=\left[K_{jl}\right]_{j,l=1}^{N}=\left[\begin{array}{cccc}
u_{\infty}(\vv_1,\vt_1;\omega) & u_{\infty}(\vv_1,\vt_2;\omega) & \cdots & u_{\infty}(\vv_1,\vt_N;\omega)\\
u_{\infty}(\vv_2,\vt_1;\omega) & u_{\infty}(\vv_2,\vt_2;\omega) & \cdots & u_{\infty}(\vv_2,\vt_N;\omega)\\
\vdots&\vdots&\ddots&\vdots\\
u_{\infty}(\vv_N,\vt_1;\omega) & u_{\infty}(\vv_N,\vt_2;\omega) & \cdots & u_{\infty}(\vv_N,\vt_N;\omega)\\
\end{array}\right].
\end{equation}
Because the incident and observation directions are the same, $K_{jl}$ becomes
\begin{align}
\begin{aligned}\label{StructureofMSRmatrix}
  K_{jl}=&u_{\infty}(-\vt_j,\vt_l;\omega)\\
  &\approx h\frac{\omega^2(1+i)}{4\sqrt{\omega\pi}} \int_\gamma\bigg(\vt_j\cdot\mathbb{M}(\mx)\cdot\vt_l+(\eps-1)\bigg)e^{i\omega(\vt_j+\vt_l)\cdot\mx}d\gamma(\mx)\\
  =&h\frac{\omega^2(1+i)}{4\sqrt{\omega\pi}}\frac{\mbox{length}(\gamma)}{M} \sum_{m=1}^{M}\bigg[2\left(\frac{1}{\mu}-1\right)\vt_j\cdot\mt(\mx_m)\vt_l\cdot\mt(\mx_m)\\
  &+2(1-\mu)\vt_j\cdot\mn(\mx_m)\vt_l\cdot\mn(\mx_m)+(\eps-1)\bigg]e^{i\omega(\vt_j+\vt_l)\cdot\mx}d\gamma(\mx),
\end{aligned}
\end{align}
where $\mbox{length}(\gamma)$ denotes the length of $\gamma$ (refer to \cite{PL3}).

With this background, the MUSIC algorithm can be introduced as follows. Let us perform the singular value decomposition of $\mathbb{K}$:
\[\mathbb{K}=\mathbb{USV}^*\approx\sum_{m=1}^M\sigma_m\mU_m\mV_m^*,\]
where $\mU_m$ and $\mV_m$ are the left and right singular vectors of $\mathbb{K}$, respectively, and $\sigma_m$ denotes the non-zero singular values. Then, on the basis of the structure of the MSR matrix (\ref{StructureofMSRmatrix}), we define a vector $\mf(\mz)\in\mathbb{C}^{N\times1}$ as
\begin{equation}\label{Vecf}
  \mf(\mz)=\bigg[\mathbf{c}_1\cdot[1,\vt_1]^Te^{i\omega\vt_1\cdot\mz}, \mathbf{c}_2\cdot[1,\vt_2]^Te^{i\omega\vt_2\cdot\mz},\cdots, \mathbf{c}_N\cdot[1,\vt_N]^Te^{i\omega\vt_N\cdot\mz}\bigg]^T,
\end{equation}
where the selection of $\mathbf{c}_n\in\mathbb{R}^3\backslash\set{\mathbf{0}}$, $n=1,2,\cdots,N$, depends on the shape of the supporting curve $\gamma(\mx)$. In fact, this is a linear combination of $\mt(\mx_m)$ and $\mn(\mx_m)$.

Let us define a projection operator onto the noise subspace:
\[P_{\mathrm{noise}}(\mf(\mz)):=\left(\mathbb{I}_{N}-\sum_{m=1}^{M}\mU_m\mU_m^*\right)\mf(\mz).\]
Then, the MUSIC-type imaging functional can be introduced:
\begin{equation}\label{MUSICfunction}
  \mathbb{E}(\mz):=\frac{1}{|P_{\mathrm{noise}}(\mf(\mz))|}.
\end{equation}
Note that $\mathbb{K}$ is symmetric, but it is not self-adjoint. Therefore, we form a self-adjoint matrix
\[\mathbb{A}=\mathbb{K^*K}=\mathbb{\overline{K}K},\]
where $*$ denotes the adjoint, and the bar denotes the complex conjugate. Then, with a careful choice of $\mathbf{c}_n$, the range of $\mathbb{A}$ is spanned by the $M$ vectors $\set{\mf(\mx_1),\mf(\mx_2),\cdots,\mf(\mx_M)}$ (see \cite{AK,C} for instance). Therefore,
\[\mf(\mz)\in\mbox{Range}(\mathbb{\overline{K}K})\quad\mbox{if and only if}\quad\mz\in\set{\mx_1,\mx_2,\cdots,\mx_M};\]
i.e., equivalently $|P_{\mathrm{noise}}(\mf(\mz))|=0$. Thus, the map of (\ref{MUSICfunction}) will show a large magnitude (theoretically, $+\infty$) at $\mx_m\in\gamma$.

\begin{rem}
  The results of the numerical simulations in \cite{P1,PL1,PL3} indicate that the selection of $\mathbf{c}_n$ is a strong prerequisite. The selection depends on the shape of the supporting curve $\gamma$. For purely dielectric contrast, $\mathbf{c}_n=[1,0,0]^T$ is a good choice. However, for purely magnetic contrast, $\mathbf{c}_n$ must be in the form $\mathbf{c}_n=[0,\mathbf{b}]^T$, where $\mathbf{b}$ is a linear combination of $\mt(\mx_m)$ and $\mn(\mx_m)$ for $m=1,2,\cdots,M$. Unfortunately, we have no \textit{a priori} information on the shape of $\gamma$. Therefore, in \cite{HSZ,PL1}, a  large number of directions are applied to find an optimal vector $\mathbf{b}$. Applying this $\mathbf{b}$ yields a good result, but this process incurs large computational costs. Hence, motivated by recent work \cite{HSZ}, we assume that $\mathbf{c}_n$ satisfies $\mathbf{c}_n\cdot[1,\vt_n]^T=1$ for all $n$, i.e.,
  \begin{equation}\label{VecF}
    \mf(\mz)=\bigg[e^{i\omega\vt_1\cdot\mz},e^{i\omega\vt_2\cdot\mz},\cdots,e^{i\omega\vt_N\cdot\mz}\bigg]^T,
  \end{equation}
  and explore some properties of the MUSIC-type imaging algorithm.
\end{rem}

\section{Structure of certain properties of MUSIC-type imaging function}\label{sec:3}
In this section, we identify the structure of the MUSIC-type imaging function. Before starting, we recall a useful result derived in \cite{G}. This plays a key role in our identification of the structure.
\begin{lem}\label{TheoremBessel}
  Assume that $\set{\vt_n:n=1,2,\cdots,N}$ spans unit circle $\mathbb{S}^1$. Then, the following identities hold for sufficiently large $N$ and $\vx,\mx\in\mathbb{R}^2$.
\begin{align*}
  &\frac{1}{N}\sum_{n=1}^{N}e^{i\omega\vt_n\cdot\mx}=\frac{1}{2\pi}\int_{\mathbb{S}^1}e^{i\omega\vt\cdot\mx}dS(\vt)=J_0(\omega|\mx|),\\
  &\frac{1}{N}\sum_{n=1}^{N}(\vt_n\cdot\vx)e^{i\omega\vt_n\cdot\mx}=\frac{1}{2\pi}\int_{\mathbb{S}^1}(\vt\cdot\vx)e^{i\omega\vt\cdot\mx}dS(\vt)=i\left(\frac{\mx}{|\mx|}\cdot\vx\right)J_1(\omega|\mx|),
\end{align*}
where $J_p$ denotes the Bessel function of integer order $p$ of the first kind.
\end{lem}

\subsection{Pure dielectric permittivity contrast case: $\eps\ne\eps_0$ and $\mu=\mu_0$}
First, we consider the dielectric permittivity contrast case; i.e., we assume that $\eps\ne\eps_0$ and $\mu=\mu_0$. The proof is similar to the result in \cite[Theorem 3.3]{JP}.
\begin{thm}[Pure dielectric permittivity contrast case]\label{Theorem1}
  For sufficiently large $N$ ($>M$) and $\omega$, (\ref{MUSICfunction}) can be written as follows:
  \begin{equation}\label{MUSIC1}
    \mathbb{E}_{\eps}(\mz):=\frac{1}{|P_{\mathrm{ noise}}(\mf(\mz))|}\approx\frac{1}{\sqrt{N}}\left(1-\sum_{m=1}^{M}J_0(\omega|\mz-\mx_m|)^2+o(h)\right)^{-1/2}.
  \end{equation}
\end{thm}
\begin{proof}
In this case, if we let $\mu=\mu_0$ in (\ref{StructureofMSRmatrix}), the left singular vectors are of the form
\[\mU_m\approx\frac{1}{\sqrt{N}}\bigg[e^{i\omega\vt_1\cdot \mx_m},e^{i\omega\vt_2\cdot \mx_m},\cdots,e^{i\omega\vt_N\cdot \mx_m}\bigg]^T+\mO(h),\]
where $\mO(h)$ denotes a $N\times1$ vector whose elements are $o(h)$. Then, for sufficiently large $\omega$, we can observe that
\begin{equation}\label{orthogonal1}
  \mU_m\cdot\overline{\mU}_{m'}\approx\frac{1}{N}\sum_{n=1}^{N}e^{i\omega\vt_n\cdot(\mx_m-\mx_{m'})}\approx J_0(\omega|\mx_m-\mx_{m'}|)+o(h)\approx\left\{\begin{array}{ccl}
1&\mbox{if}&m=m'\\
\noalign{\medskip}0&\mbox{if}&m\ne m'.
\end{array}\right.
\end{equation}

Following elementary calculus, $P_{\mathrm{noise}}$ can be written as
\begin{align*}
  P_{\mathrm{noise}}(\mf(\mz))&=\left(\mathbb{I}_{N}-\sum_{m=1}^{M}\mU_m\overline{\mU}_m^T\right)\mf(\mz)\\
  &=\left[\begin{array}{c}
        e^{i\omega\vt_1\cdot\mz} \\
        e^{i\omega\vt_2\cdot\mz} \\
        \vdots \\
        e^{i\omega\vt_N\cdot\mz} \\
      \end{array}\right]
  -\frac{1}{N}\sum_{m=1}^{M}\left[\begin{array}{c}
      \displaystyle e^{i\omega\vt_1\cdot\mz}+\sum_{n\in\mathbb{N}_1}e^{i\omega\vt_1\cdot \mx_m}e^{i\omega\vt_n\cdot(\mz-\mx_m)}+o(h)\\
      \displaystyle e^{i\omega\vt_2\cdot\mz}+\sum_{n\in\mathbb{N}_2}e^{i\omega\vt_1\cdot \mx_m}e^{i\omega\vt_n\cdot(\mz-\mx_m)}+o(h)\\
      \vdots \\
      \displaystyle e^{i\omega\vt_N\cdot\mz}+\sum_{n\in\mathbb{N}_N}e^{i\omega\vt_1\cdot \mx_m}e^{i\omega\vt_n\cdot(\mz-\mx_m)}+o(h)
      \end{array}\right],
\end{align*}
where $\mathbb{N}_n=\{1,2,\cdots,N\}\backslash\{n\}$. Because
\[e^{i\omega\vt_n\cdot\mz}=e^{i\omega\vt_n\cdot \mx_m}e^{i\omega\vt_n\cdot (\mz-\mx_m)},\] $P_{\mathrm{noise}}$ can be expressed as
\[P_{\mathrm{noise}}(\mf(\mz))=\left[
      \begin{array}{c}
        \displaystyle e^{i\omega\vt_1\cdot\mz}-\sum_{m=1}^{M}e^{i\omega\vt_1\cdot \mx_m}J_0(\omega|\mz-\mx_m|)+o(h) \\
        \displaystyle e^{i\omega\vt_2\cdot\mz}-\sum_{m=1}^{M}e^{i\omega\vt_2\cdot \mx_m}J_0(\omega|\mz-\mx_m|)+o(h) \\
        \vdots \\
        \displaystyle e^{i\omega\vt_N\cdot\mz}-\sum_{m=1}^{M}e^{i\omega\vt_N\cdot \mx_m}J_0(\omega|\mz-\mx_m|)+o(h) \\
      \end{array}
    \right].\]
Therefore, we can obtain
\[|P_{\mathrm{noise}}(\mf(\mz))|=\bigg(P_{\mathrm{noise}}(\mf(\mz))\cdot\overline{P_{\mathrm{noise}}(\mf(\mz))}\bigg)^{1/2} =\left(\sum_{n=1}^{N}\bigg(1-\Phi_1+\Phi_2+o(h)\bigg)\right)^{1/2},\]
where
\begin{align*}
  \Phi_1&=\sum_{m=1}^{M}\bigg(e^{i\omega\vt_n\cdot(\mz-\mx_m)}+e^{-i\omega\vt_n\cdot(\mz-\mx_m)}\bigg)J_0(\omega|\mz-\mx_m|)\\
  \Phi_2&=\left(\sum_{m=1}^{M}e^{i\omega\vt_n\cdot\mx_m}J_0(\omega|\mz-\mx_m|)\right) \left(\sum_{m=1}^{M}e^{-i\omega\vt_n\cdot\mx_m}J_0(\omega|\mz-\mx_m|)\right).
\end{align*}
Because
\begin{align}
\begin{aligned}\label{term1}
  \sum_{n=1}^{N}\Phi_1&=\sum_{n=1}^{N}\sum_{m=1}^{M}\bigg(e^{i\omega\vt_n\cdot(\mz-\mx_m)}+e^{-i\omega\vt_n\cdot(\mz-\mx_m)}\bigg)J_0(\omega|\mz-\mx_m|)\\ &=2N\sum_{m=1}^{M}J_0(\omega|\mz-\mx_m|)^2,
\end{aligned}
\end{align}
and on the basis of the orthogonal property (\ref{orthogonal1}), we can evaluate
\begin{align}
\begin{aligned}\label{term2}
  \sum_{n=1}^{N}\Phi_2&=\left(\sum_{m=1}^{M}e^{i\omega\vt_n\cdot\mx_m}J_0(\omega|\mz-\mx_m|)\right) \left(\sum_{m=1}^{M}e^{-i\omega\vt_n\cdot\mx_m}J_0(\omega|\mz-\mx_m|)\right)\\
  &=\sum_{n=1}^{N}\sum_{m=1}^{M}\sum_{m'=1}^{M}\bigg(e^{-i\omega\vt_n\cdot\mx_m}J_0(\omega|\mz-\mx_m|)\bigg) \bigg(e^{i\omega\vt_n\cdot\mx_{m'}}J_0(\omega|\mz-\mx_{m'}|)\bigg)\\
  &=\sum_{m=1}^{M}\sum_{m'=1}^{M}\sum_{n=1}^{N}\bigg(e^{i\omega\vt_n\cdot(\mx_{m'}-\mx_m)}J_0(\omega|\mz-\mx_m|)J_0(\omega|\mz-\mx_{m'}|)\bigg)\\
  &=N\sum_{m=1}^{M}\sum_{m'=1}^{M}J_0(\omega|\mx_{m'}-\mx_m|)J_0(\omega|\mz-\mx_m|)J_0(\omega|\mz-\mx_{m'}|)\\
  &=N\sum_{m=1}^{M}J_0(\omega|\mz-\mx_m|)^2.
\end{aligned}
\end{align}
Therefore, using (\ref{term1}) and (\ref{term2}), we can obtain the following result:
\[|P_{\mathrm{noise}}(\mf(\mz))|=\sqrt{N}\left(1-\sum_{m=1}^{M}J_0(\omega|\mz-\mx_m|)^2+o(h)\right)^{1/2}.\]
\end{proof}

Note that $J_0(x)$ has its maximum value $1$ at $x=0$. This means that plots of $\mathbb{E}_{\eps}(\mz)$ will show peaks of large ($+\infty$ in theory) and small magnitude at $\mx_m\in\gamma$ and at $\mx\notin\gamma$, respectively (see FIG. \ref{PlotMUSIC1}). This is why the MUSIC algorithm offers a good result for the pure dielectric contrast case of the full-view inverse scattering problem. We refer to FIG. \ref{MapMUSIC1} and various results in \cite{AKLP,HSZ,PL1,PL3}.

\subsection{Pure magnetic permeability contrast case: $\eps=\eps_0$ and $\mu\ne\mu_0$}
Next, we consider the magnetic permeability contrast case; i.e., we assume that $\eps=\eps_0$ and $\mu\ne\mu_0$. The result is as follows.
\begin{thm}[Pure magnetic permeability contrast case]\label{Theorem2}
  For sufficiently large $N$ ($>2M$) and $\omega$, (\ref{MUSICfunction}) can be written as follows:
  \begin{equation}\label{MUSIC2}
    \mathbb{E}_{\mu}(\mz)=\frac{1}{\sqrt{N}}\left(1-\sum_{m=1}^{M}\left(\frac{\mz-\mx_m}{|\mz-\mx_m|}\cdot(\mt(\mx_m)+\mn(\mx_m))\right)^2J_1(\omega|\mz-\mx_m|)^2+o(h)\right)^{-1/2}.
  \end{equation}
\end{thm}
\begin{proof}
In this case, the left singular vectors are of the form (see \cite{AGKPS,PL3})
\[\mU_{2(m-1)+s}\approx\frac{1}{\sqrt{N}}\bigg[\vt_1\cdot\vx_s(\mx_m)e^{i\omega\vt_1\cdot\mx_m},
%\vt_2\cdot\vx_s(\mx_m)e^{i\omega\vt_2\cdot\mx_m},
\cdots,\vt_N\cdot\vx_s(\mx_m)e^{i\omega\vt_N\cdot\mx_m}\bigg]^T+\mO(h),\]
where
\[\vx_s(\mx_m):=\left\{\begin{array}{rcl}
                         \mt(\mx_m) & \mbox{if} & s=1 \\
                         \noalign{\medskip}\mn(\mx_m) & \mbox{if} & s=2.
                       \end{array}
\right.\]
Then, based on the orthonormal property of singular vectors, we can observe the following: because $\omega$ is sufficiently large, if $m\ne m'$ or $s\ne s''$, then
\begin{multline}\label{orthogonal2}
  \mU_{2(m-1)+s}\cdot\overline{\mU}_{2(m'-1)+s''}\approx\frac{1}{N} \sum_{n=1}^{N}\bigg(\vt_n\cdot(\vx_s(\mx_m)+\vx_{s''}(\mx_{m'})\bigg) e^{i\omega\vt_n\cdot(\mx_m-\mx_{m'})}\\
  \approx i\frac{\mx_m-\mx_{m'}}{|\mx_m-\mx_{m'}|}\cdot\bigg(\vx_s(\mx_m)+\vx_{s''}(\mx_{m'})\bigg)J_1(\omega|\mx_m-\mx_{m'}|)+o(h)\approx0.
\end{multline}

Because we selected $\mf(\mz)$ as (\ref{VecF}), $P_{\mathrm{noise}}$ can be written as
\begin{align*}
  &P_{\mathrm{noise}}(\mf(\mz)) =\left(\mathbb{I}_{N}-\sum_{m=1}^{M}\sum_{s=1}^{2}\mU_{2(m-1)+s}\overline{\mU}_{2(m-1)+s}^T\right)\mf(\mz)\\
  &\approx\left[\begin{array}{c}
        e^{i\omega\vt_1\cdot\mz} \\
        e^{i\omega\vt_2\cdot\mz} \\
        \vdots \\
        e^{i\omega\vt_N\cdot\mz} \\
      \end{array}\right]-
      \frac{1}{N}\sum_{m=1}^{M}\sum_{s=1}^{2}\left[\begin{array}{c}
      \displaystyle(\vt_1\cdot\vx_s(\mx_m))e^{i\omega\vt_1\cdot\mx_m}\sum_{n=1}^{N}(\vt_n\cdot\vx_s(\mx_m))e^{i\omega\vt_n\cdot(\mz-\mx_m)}+o(h)\\
      \displaystyle(\vt_2\cdot\vx_s(\mx_m))e^{i\omega\vt_2\cdot\mx_m}\sum_{n=1}^{N}(\vt_n\cdot\vx_s(\mx_m))e^{i\omega\vt_n\cdot(\mz-\mx_m)}+o(h)\\
      \vdots\\
      \displaystyle(\vt_N\cdot\vx_s(\mx_m))e^{i\omega\vt_N\cdot\mx_m}\sum_{n=1}^{N}(\vt_n\cdot\vx_s(\mx_m))e^{i\omega\vt_n\cdot(\mz-\mx_m)}+o(h)
      \end{array}\right]\\
      &=\left[\begin{array}{c}
      \displaystyle e^{i\omega\vt_1\cdot\mz}-i\sum_{m=1}^{M}\sum_{s=1}^{2}(\vt_1\cdot\vx_s(\mx_m))\left(\frac{\mz-\mx_m}{|\mz-\mx_m|}\cdot\vx_s(\mx_m)\right)e^{i\omega\vt_1\cdot \mx_m}J_1(\omega|\mz-\mx_m|)+o(h)\\
      \displaystyle e^{i\omega\vt_2\cdot\mz}-i\sum_{m=1}^{M}\sum_{s=1}^{2}(\vt_2\cdot\vx_s(\mx_m))\left(\frac{\mz-\mx_m}{|\mz-\mx_m|}\cdot\vx_s(\mx_m)\right)e^{i\omega\vt_2\cdot \mx_m}J_1(\omega|\mz-\mx_m|)+o(h)\\
      \vdots\\
      \displaystyle e^{i\omega\vt_N\cdot\mz}-i\sum_{m=1}^{M}\sum_{s=1}^{2}(\vt_N\cdot\vx_s(\mx_m))\left(\frac{\mz-\mx_m}{|\mz-\mx_m|}\cdot\vx_s(\mx_m)\right)e^{i\omega\vt_N\cdot \mx_m}J_1(\omega|\mz-\mx_m|)+o(h)
      \end{array}\right].
\end{align*}
Hence,
\[|P_{\mathrm{noise}}(\mf(\mz))|=\left(\sum_{n=1}^{N}\bigg(1+\Psi_1-\overline{\Psi}_1+\Psi_2\overline{\Psi}_2+o(h)\bigg)\right)^{1/2},\]
where
\begin{align*}
  \Psi_1&=i\sum_{m=1}^{M}\sum_{s=1}^{2}(\vt_n\cdot\vx_s(\mx_m))\left(\frac{\mz-\mx_m}{|\mz-\mx_m|}\cdot\vx_s(\mx_m)\right) e^{i\omega\vt_n\cdot(\mz-\mx_m)}J_1(\omega|\mz-\mx_m|)\\
  %\Psi_2&=i\sum_{m=1}^{M}\sum_{s=1}^{2}(\vt_n\cdot\vx_s(\mx_m))\left(\frac{\mz-\mx_m}{|\mz-\mx_m|}\cdot\vx_s(\mx_m)\right) e^{-i\omega\vt_n\cdot(\mz-\mx_m)}J_1(\omega|\mz-\mx_m|)\\
  \Psi_2&=\sum_{m=1}^{M}\sum_{s=1}^{2}(\vt_n\cdot\vx_s(\mx_m))\left(\frac{\mz-\mx_m}{|\mz-\mx_m|}\cdot\vx_s(\mx_m)\right)e^{i\omega\vt_n\cdot \mx_m}J_1(\omega|\mz-\mx_m|).
  %\Psi_4&=\sum_{m=1}^{M}\sum_{s=1}^{2}(\vt_n\cdot\vx_s(\mx_m))\left(\frac{\mz-\mx_m}{|\mz-\mx_m|}\cdot\vx_s(\mx_m)\right)e^{-i\omega\vt_n\cdot \mx_m}J_1(\omega|\mz-\mx_m|).
\end{align*}

Because
\begin{align*}
  &\sum_{n=1}^{N}(\Psi_1-\overline{\Psi}_1)\\
  =&i\sum_{n=1}^{N}\sum_{m=1}^{M}\sum_{s=1}^{2}(\vt_n\cdot\vx_s(\mx_m))\left(\frac{\mz-\mx_m}{|\mz-\mx_m|}\cdot\vx_s(\mx_m)\right) e^{i\omega\vt_n\cdot(\mz-\mx_m)}J_1(\omega|\mz-\mx_m|)\\
  &+i\sum_{n=1}^{N}\sum_{m=1}^{M}\sum_{s=1}^{2}(\vt_n\cdot\vx_s(\mx_m))\left(\frac{\mz-\mx_m}{|\mz-\mx_m|}\cdot\vx_s(\mx_m)\right) e^{-i\omega\vt_n\cdot(\mz-\mx_m)}J_1(\omega|\mz-\mx_m|)\\
  =&\sum_{m=1}^{M}\sum_{s=1}^{2}\left(\frac{\mz-\mx_m}{|\mz-\mx_m|}\cdot\vx_s(\mx_m)\right)\sum_{n=1}^{N}\bigg(i(\vt_n\cdot\vx_s(\mx_m))e^{i\omega\vt_n\cdot(\mz-\mx_m)}\bigg)J_1(\omega|\mz-\mx_m|)\\
  &+\sum_{m=1}^{M}\sum_{s=1}^{2}\left(\frac{\mz-\mx_m}{|\mz-\mx_m|}\cdot\vx_s(\mx_m)\right)\sum_{n=1}^{N}\bigg(i(\vt_n\cdot\vx_s(\mx_m))e^{-i\omega\vt_n\cdot(\mz-\mx_m)}\bigg)J_1(\omega|\mz-\mx_m|),
%  =&-2N\sum_{m=1}^{M}\sum_{s=1}^{2}\bigg(\frac{\mz-\mx_m}{|\mz-\mx_m|}\cdot\vx_s(\mx_m)\bigg)^2J_1(\omega|\mz-\mx_m|)^2\\
%  =&-2N\sum_{m=1}^{M}\bigg(\frac{\mz-\mx_m}{|\mz-\mx_m|}\cdot\bigg(\mt(\mx_m)+\mn(\mx_m)\bigg)\bigg)^2J_1(\omega|\mz-\mx_m|)^2.
\end{align*}
applying Theorem \ref{TheoremBessel}, we can obtain
\begin{align}
  \begin{aligned}\label{term3}
    \sum_{n=1}^{N}(\Psi_1-\overline{\Psi}_1)&=-2N\sum_{m=1}^{M}\sum_{s=1}^{2}\bigg(\frac{\mz-\mx_m}{|\mz-\mx_m|}\cdot\vx_s(\mx_m)\bigg)^2J_1(\omega|\mz-\mx_m|)^2\\
  &=-2N\sum_{m=1}^{M}\bigg(\frac{\mz-\mx_m}{|\mz-\mx_m|}\cdot\bigg(\mt(\mx_m)+\mn(\mx_m)\bigg)\bigg)^2J_1(\omega|\mz-\mx_m|)^2.
  \end{aligned}
\end{align}

Next, on the basis of the orthogonal property (\ref{orthogonal2}), we can evaluate
\begin{align*}
  \sum_{n=1}^{N}\Psi_2\overline{\Psi}_2=&\sum_{n=1}^{N}\left(\sum_{m=1}^{M}\sum_{s=1}^{2}(\vt_n\cdot\vx_s(\mx_m))\left(\frac{\mz-\mx_m}{|\mz-\mx_m|}\cdot\vx_s(\mx_m)\right)e^{i\omega\vt_n\cdot \mx_m}J_1(\omega|\mz-\mx_m|)\right)\\
  &\times\left(\sum_{m'=1}^{M}\sum_{s''=1}^{2}(\vt_n\cdot\vx_{s''}(\mx_{m'}))\left(\frac{\mz-\mx_{m'}}{|\mz-\mx_{m'}|}\cdot\vx_{s''}(\mx_{m'})\right)e^{-i\omega\vt_n\cdot\mx_{m'}}J_1(\omega|\mz-\mx_{m'}|)\right)\\
  %=&\sum_{n=1}^{N}\left(\sum_{m=1}^{M}(\vt_n\cdot\vx_s(\mx_m))\left(\frac{\mz-\mx_m}{|\mz-\mx_m|}\cdot\vx_s(\mx_m)\right)e^{i\omega\vt_n\cdot \mx_m}J_1(\omega|\mz-\mx_m|)\right)\\
  %&\times\left(\sum_{m=1}^{M}(\vt_n\cdot\vx_{s''}(\mx_m))\left(\frac{\mz-\mx_m}{|\mz-\mx_m|}\cdot\vx_{s''}(\mx_m)\right)e^{-i\omega\vt_n\cdot \mx_m}J_1(\omega|\mz-\mx_m|)\right)\\
  =&\sum_{m=1}^{M}\sum_{n=1}^{N}\left(\sum_{s=1}^{2}(\vt_n\cdot\vx_s(\mx_m))\left(\frac{\mz-\mx_m}{|\mz-\mx_m|}\cdot\vx_s(\mx_m)\right)e^{i\omega\vt_n\cdot \mx_m}J_1(\omega|\mz-\mx_m|)\right)\\
  &\times\left(\sum_{{s''}=1}^{2}(\vt_n\cdot\vx_{s''}(\mx_{m'}))\left(\frac{\mz-\mx_m}{|\mz-\mx_m|}\cdot\vx_{s''}(\mx_m)\right)e^{-i\omega\vt_n\cdot \mx_{m}}J_1(\omega|\mz-\mx_{m}|)\right)\\
  =&2\sum_{m=1}^{M}\sum_{n=1}^{N}\left((\vt_n\cdot\vx_s(\mx_m))^2\left(\frac{\mz-\mx_m}{|\mz-\mx_m|}\cdot\vx_s(\mx_m)\right)^2J_1(\omega|\mz-\mx_m|)^2\right)\\
  =&2N\sum_{m=1}^{M}\sum_{s=1}^{2}\left(\frac{1}{N}\sum_{n=1}^{N}(\vt_n\cdot\vx_s(\mx_m))^2\right)\left(\frac{\mz-\mx_m}{|\mz-\mx_m|}\cdot\vx_s(\mx)\right)^2J_1(\omega|\mz-\mx_m|)^2\\
  =&\frac{N}{\pi}\sum_{m=1}^{M}\sum_{s=1}^{2}\left(\frac{\mz-\mx_m}{|\mz-\mx_m|}\cdot\vx_s(\mx)\right)^2\int_{\mathbb{S}^1}(\vp\cdot\vx_s(\mx_m))^2dS(\vp)J_1(\omega|\mz-\mx_m|)^2.
\end{align*}
Now, we consider polar coordinates; because $\vp,\vx_s\in\mathbb{S}^1$, let $\vp=[\cos\phi,\sin\phi]^T$ and $\vx_s=[\cos\psi,\sin\psi]^T$; then elementary calculus yields
\[\int_{\mathbb{S}^1}(\vp\cdot\vx_s(\mx_m))^2dS(\vp)=\int_0^{2\pi}\cos^2(\phi-\psi)d\phi=\pi.\]
%2\int_0^{\pi}\cos^2\phi d\phi=2\left[\frac12\sin\phi\cos\phi+\frac12\phi\right]_0^{\pi}=
Hence, we can obtain
\begin{equation}\label{term4}
  \sum_{n=1}^{N}\Psi_2\overline{\Psi}_2=N\sum_{m=1}^{M}\left(\frac{\mz-\mx_m}{|\mz-\mx_m|}\cdot\bigg(\mt(\mx_m)+\mn(\mx_m)\bigg)\right)^2J_1(\omega|\mz-\mx_m|)^2.
\end{equation}

Therefore, from (\ref{term3}) and (\ref{term4}), we can obtain
\[|P_{\mathrm{ noise}}(\mf(\mz))|=\sqrt{N}\left(1-\sum_{m=1}^{M}\left(\frac{\mz-\mx_m}{|\mz-\mx_m|}\cdot(\mt(\mx_m)+\mn(\mx_m))\right)^2J_1(\omega|\mz-\mx_m|)^2+o(h)\right)^{1/2}.\]
\end{proof}

Unlike the permittivity contrast case, the map of (\ref{MUSIC2}) shows two curves in the neighborhood of $\gamma$ because
\[\lim_{\mz\to\mx_m}\frac{J_1(\omega|\mz-\mx_m|)^2}{|\mz-\mx_m|}=0,\]
and $J_1(x)^2$ is maximum at two points, $x_1$ and $x_2$, which are symmetric with respect to $x=0$. This is why two ghost replicas with large magnitude and many artifacts with small magnitude appear instead of the true shape of the supporting curve $\gamma$ (see FIG. \ref{PlotMUSIC2}). Some numerical simulation results can be found in FIG. \ref{MapMUSIC2} and in \cite[Section 5]{HSZ}.

Note that $J_1(x)^2\ne1$ for all $x\in\mathbb{R}$. Hence, in contrast to the permittivity contrast case, (\ref{MUSIC2}) does not blow up.

\begin{figure}[!ht]
\begin{center}
\subfigure[Graph of $|1-J_0(\omega|x|)^2|^{-1}$]{\label{PlotMUSIC1}\includegraphics[width=0.49\textwidth]{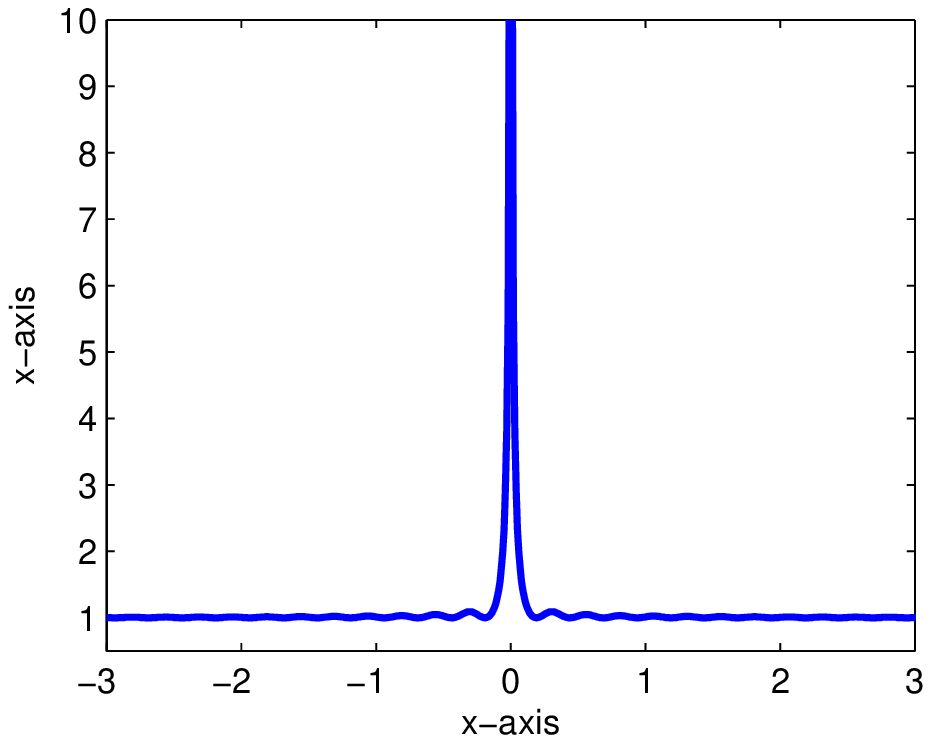}}
\subfigure[Graph of $|1-J_1(\omega|x|)^2|^{-1}$]{\label{PlotMUSIC2}\includegraphics[width=0.49\textwidth]{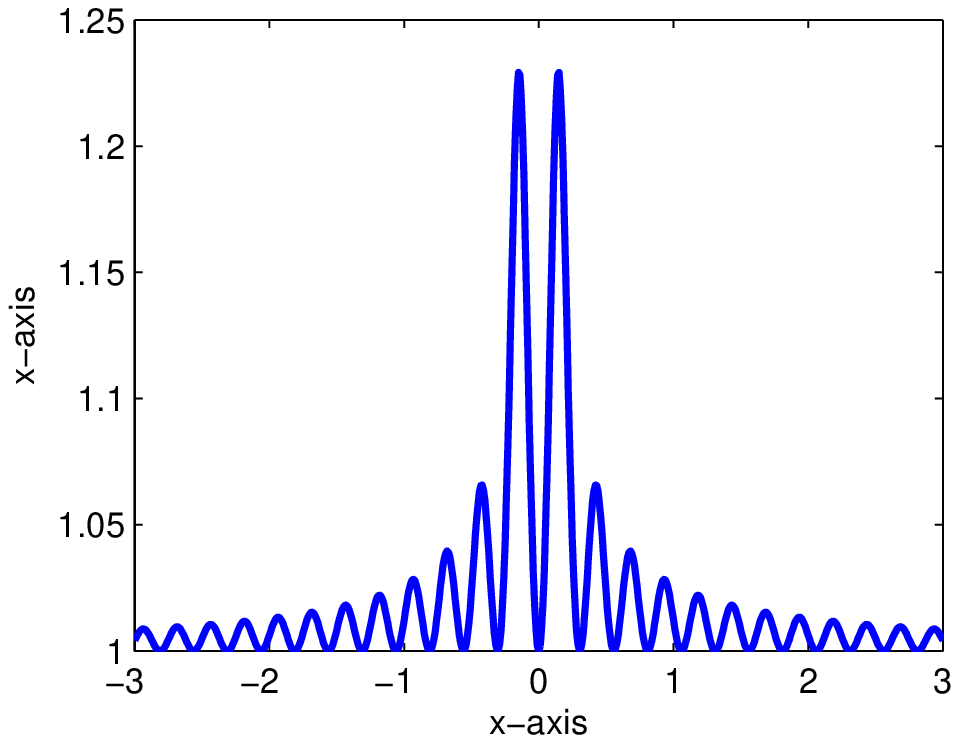}}
\caption{\label{PlotMUSIC}Graphs of $|1-J_p(\omega|x|)^2|^{-1}$, with $p=0$ (left) and $1$ (right), for $\omega=2\pi/0.5$.}
\end{center}
\end{figure}

Finally, by combining (\ref{MUSIC1}) and (\ref{MUSIC2}), we can immediately obtain the following result.

\begin{thm}[Both permittivity and permeability contrast case]\label{Theorem3}
  Let $\eps\ne\eps_0$ and $\mu\ne\mu_0$. Then, for sufficiently large $N$ ($>3M$) and $\omega$, (\ref{MUSICfunction}) can be written as follows:
  \begin{multline}\label{MUSIC3}
    \mathbb{E}_{\eps,\mu}(\mz):=\frac{1}{\sqrt{N}}\left(1-\sum_{m=1}^{M}J_0(\omega|\mz-\mx_m|)^2\right.\\
     \left.+\sum_{m=1}^{M}\bigg(\frac{\mz-\mx_m}{|\mz-\mx_m|}\cdot(\mt(\mx_m)+\mn(\mx_m))\bigg)^2J_1(\omega|\mz-\mx_m|)^2+o(h)\bigg)\right)^{-1/2}.
  \end{multline}
\end{thm}

This result shows that plots of (\ref{MUSIC3}) show a large magnitude at $\mz$ if $\mz\ne\mx_m$ and
\[\sum_{m=1}^{M}\left(J_0(\omega|\mz-\mx_m|)^2+\left(\frac{\mz-\mx_m}{|\mz-\mx_m|}\cdot(\mt(\mx_m)+\mn(\mx_m))\right)^2J_1(\omega|\mz-\mx_m|)^2+o(h)\right)=1.\]
Thus, a result with poor resolution will appear; refer to the examples of numerical simulation in \cite[Section 4.4]{PL3}.

\begin{rem}
  On the basis of recent work \cite{P1}, the structure of so-called \textbf{subspace migration} is as follows.
  \begin{enumerate}
    \item Permittivity contrast case:
    \[\mathbb{E}_{\mathrm{SM}}(\mz)=\sum_{m=1}^{M}J_0(\omega|\mz-\mx_m|)^2.\]
    \item Permeability contrast case:
    \[\mathbb{E}_{\mathrm{SM}}(\mz)=\sum_{m=1}^{M}\left\{\bigg(\frac{\mz-\mx_m}{|\mz-\mx_m|}\bigg)\cdot\bigg(\mt(\mx_m)+\mn(\mx_m)\bigg)J_1(\omega|\mz-\mx_m|)\right\}^2.\]
  \end{enumerate}
Hence, we can observe the following relationship between MUSIC and subspace migration, which is derived in \cite[Formula (7.4)]{AGKLS}: Let us select the unit vector $\mf(\mz)$ of (\ref{VecF}) as
\[\mf(\mz)=\frac{1}{\sqrt{N}}\bigg[e^{i\omega\vt_1\cdot\mz},e^{i\omega\vt_2\cdot\mz},\cdots,e^{i\omega\vt_N\cdot\mz}\bigg]^T.\]
Then,
\[\mathbb{E}(\mz)=\bigg(1-\mathbb{E}_{\mathrm{SM}}(\mz)\bigg)^{-1/2}.\]
\end{rem}

\subsection{Imaging of perfectly conducting cracks}
Here, let $\Gamma$ be a smooth curve that describes the crack: for an injective piecewise smooth function $\boldsymbol{\phi}:[-1,1]\longrightarrow\mathbb{R}^2$,
\begin{equation}\label{PCrack}
  \Gamma=\{\boldsymbol{\phi}(x):-1\leq x\leq1\}.
\end{equation}
Let $u(\mx,\vt;\omega)$ be the single-component electric field that satisfies the Helmholtz equation:
\[\Delta u(\mx,\vt;\omega)+\omega^2 u(\mx,\vt;\omega)=0\quad\mbox{in}\quad\mathbb{R}^2\backslash\Gamma.\]
For the sound-soft arc [transverse magnetic (TM) polarization], $u(\mx,\vt;\omega)$ satisfies the Dirichlet boundary condition on $\Gamma$ (see \cite{K}):
\[u(\mx,\vt;\omega)=0\quad\mbox{on}\quad\Gamma\]
and for the sound-hard arc [transverse electric (TE) polarization], $u(\mx,\vt;\omega)$ satisfies the Neumann boundary condition on $\Gamma$ (see \cite{M2}):
\[\frac{\p u(\mx,\vt;\omega)}{\p\mn(\mx)}=0\quad\mbox{on}\quad\Gamma\backslash\{\boldsymbol{\phi}(-1),\boldsymbol{\phi}(1)\},\]
where $\mn(\mx)$ is the unit normal vector to $\Gamma$ at $\mx$.

Then the far-field pattern $u_{\infty}(\vv,\vt;\omega)$ for the scattering of an incident field $u_0(\mx,\vt;\omega)=e^{i\omega\vt\cdot\mx}$ from $\Gamma$ is given by
\[u_{\infty}(\vv,\vt;\omega)=\left\{
  \begin{array}{ll}
    \displaystyle\medskip\frac{1+i}{4\sqrt{\pi\omega}}\int_\Gamma e^{-i\omega\vv\cdot\mx}\varphi(\mx,\vt;\omega)d\mx & \mbox{: sound-soft arc}\\
    \displaystyle\medskip\frac{(1-i)\sqrt{\omega}}{4\sqrt{\pi}}\int_\Gamma \vv\cdot\mn(\mx)e^{-i\omega\vv\cdot\mx}\psi(\mx,\vt;\omega)d\mx & \mbox{: sound-hard arc}.
  \end{array}
  \right.\]

According to the physical factorization in \cite{HSZ,PL1}, if the incident and observation directions are the same, the left singular vector of the MSR matrix is of the form
\[\mU_m=\left\{
\begin{array}{ll}
  \medskip\bigg[e^{i\omega\vt_1\cdot\mx_m},e^{i\omega\vt_2\cdot\mx_m},\cdots,e^{i\omega\vt_N\cdot\mx_m}\bigg]^T & \mbox{: sound-soft arc} \\
  \medskip\bigg[(\vt_1\cdot\mn(\mx_m))e^{i\omega\vt_1\cdot\mx_m},\cdots,(\vt_N\cdot\mn(\mx_m))e^{i\omega\vt_N\cdot\mx_m}\bigg]^T & \mbox{: sound-hard arc}.
\end{array}
\right.\]
Thus, the structure of the left singular vectors for the sound-soft and sound-hard arcs is almost the same as in the permittivity contrast and permeability contrast cases (except for the absence of a unit tangential vector), respectively. Hence, we can obtain the following result.
\begin{thm}\label{Theorem4}
  Let $N$ and $\omega$ be sufficiently large. Then (\ref{MUSICfunction}) can be written as follows.
  \begin{enumerate}
    \item Sound-soft arc (or TM) case:
    \[\mathbb{E}_{\mathrm{TM}}(\mz)\approx\frac{1}{\sqrt{N}}\left(1-\sum_{m=1}^{M}J_0(\omega|\mz-\mx_m|)^2\right)^{-1/2}.\]
    \item Sound-hard arc (or TE) case:
    \[\mathbb{E}_{\mathrm{TE}}(\mz)\approx\frac{1}{\sqrt{N}}\left(1-\sum_{m=1}^{M}\left(\frac{\mz-\mx_m}{|\mz-\mx_m|}\cdot\mn(\mx_m)\right)^2J_1(\omega|\mz-\mx_m|)^2\right)^{-1/2}.\]
  \end{enumerate}
\end{thm}

\subsection{Imaging of small electromagnetic inclusions}
We briefly consider the use of MUSIC for imaging electromagnetic inclusions $\Sigma_m$, $m=1,2,\cdots,M$, with small diameter $r$:
\[\Sigma_m:=\mx_m+r\mathbf{B}_m,\]
where $\mathbf{B}_m$ is a simply connected smooth domain containing the origin. We assume that $\Sigma_m$ are sufficiently separate from each other, and denote the collection of such inclusions as $\Sigma$.

As in section \ref{sec:2}, we let $u(\mx,\vt;\omega)$ satisfy (\ref{TotalField}) in the presence of $\Sigma$, and $u_0(\mx,\vt;\omega)$ is given by (\ref{IncidentField}). Then, the far-field pattern can be written as the following asymptotic expansion formula (see \cite{AK}):
\begin{equation}\label{smallexpansion}
  u_{\infty}(\vv,\vt;\omega)\approx r^2\frac{\omega^2(1+i)}{4\sqrt{\omega\pi}} \sum_{m=1}^{M}|\mathbf{B}_m|\bigg((-\vv)\cdot\mathbb{A}(\mx_m)\cdot\vt+(\eps-\eps_0)\bigg)e^{i\omega(\vt-\vv)\cdot\mx_m}+o(r^2).
\end{equation}
Here, $\mathbb{A}(\mx_m)$ denotes the polarization tensor corresponding to $\Sigma_m$. Then, we can obtain the following results in a similar manner:
\begin{thm}\label{Theorem5}
  Let $N$ and $\omega$ be sufficiently large. Then (\ref{MUSICfunction}) can be written as follows.
  \begin{enumerate}
    \item Dielectric permittivity contrast case:
    \[\mathbb{E}_{\eps}(\mz)\approx\frac{1}{\sqrt{N}}\left(1-\sum_{m=1}^{M}J_0(\omega|\mz-\mx_m|)^2+o(r^2)\right)^{-1/2}.\]
    \item Magnetic permeability contrast case:
    \[\mathbb{E}_{\mu}(\mz)\approx\frac{1}{\sqrt{N}}\left(1-\sum_{m=1}^{M}\left\{\left(\frac{\mz-\mx_m}{|\mz-\mx_m|}\cdot(\me_1+\me_2)\right)J_1(\omega|\mz-\mx_m|)\right\}^2+o(r^2)\right)^{-1/2},\]
    where $\left\{\me_1=[1,0]^T,\me_2=[0,1]^T\right\}$ denotes an orthonormal basis in $\mathbb{R}^2$.
    \item Both permittivity and permeability contrast case:
    \begin{multline*}
      \mathbb{E}_{\eps,\mu}(\mz)\approx\frac{1}{\sqrt{N}}\left(1-\sum_{m=1}^{M}J_0(\omega|\mz-\mx_m|)^2\right.\\
      \left.+\sum_{m=1}^{M}\left(\frac{\mz-\mx_m}{|\mz-\mx_m|}\cdot(\me_1+\me_2)\right)^2J_1(\omega|\mz-\mx_m|)^2+o(r^2)\right)^{-1/2}.
    \end{multline*}
  \end{enumerate}
\end{thm}

\begin{rem}
  If $\Sigma_m$ denotes a perfectly conducting inclusion with a small diameter, the asymptotic expansion formulas for the TM and TE cases are similar to (\ref{smallexpansion}); refer to \cite[Theorem 3.1]{G}. Hence, the structure of (\ref{MUSICfunction}) will be the same as in Theorem \ref{Theorem5}.
\end{rem}

\section{Numerical examples}\label{sec:4}
In this section, we present some numerical simulation results. For this, we choose a thin inclusion $\Gamma_1=\{\mx+\eta\mn(x):\mx\in\gamma_1,~\eta\in(-h,h)\}$ with a smooth supporting curve:
\[\gamma_1=\left\{\left[x+0.2,x^3+x^2-0.3\right]^T:-0.5\leq x\leq0.5\right\}.\]
The thickness $h$ of the thin inclusion $\Gamma_1$ is set to $0.02$, and the following parameters are chosen: $\eps_0=\mu_0=1$, and $\eps=\mu=5$. For the illumination and observation directions $\vt_n$, they are chosen as
\[\vt_n=\left[\cos\frac{2n\pi}{N},\sin\frac{2n\pi}{N}\right]^T\]
for $n=1,2,\cdots,N$. The total number of directions is $N=24$, and the applied frequency is $\omega=2\pi/\lambda$ with a wavelength of $\lambda=0.4$. The data set for the MSR matrix $\mathbb{K}$ in (\ref{MSR}) is collected by solving the forward problem introduced in \cite{NK}.

FIG. \ref{MapMUSIC} shows maps of $\mathbb{E}_{\eps}(\mz)$ and $\mathbb{E}_{\mu}(\mz)$ in the presence of $\Gamma_1$. This result demonstrates that the MUSIC algorithm offers a very accurate result for the permittivity contrast case. For the permeability contrast case, as we saw in Theorem \ref{Theorem2}, we cannot recognize the true shape of $\Gamma_1$. However, using Theorem \ref{Theorem2}, we can obtain an approximate shape of $\Gamma_1$ from the two identified curves.

\begin{figure}[!ht]
  \begin{center}
  \subfigure[Permittivity contrast case]{\label{MapMUSIC1}\includegraphics[width=0.49\textwidth]{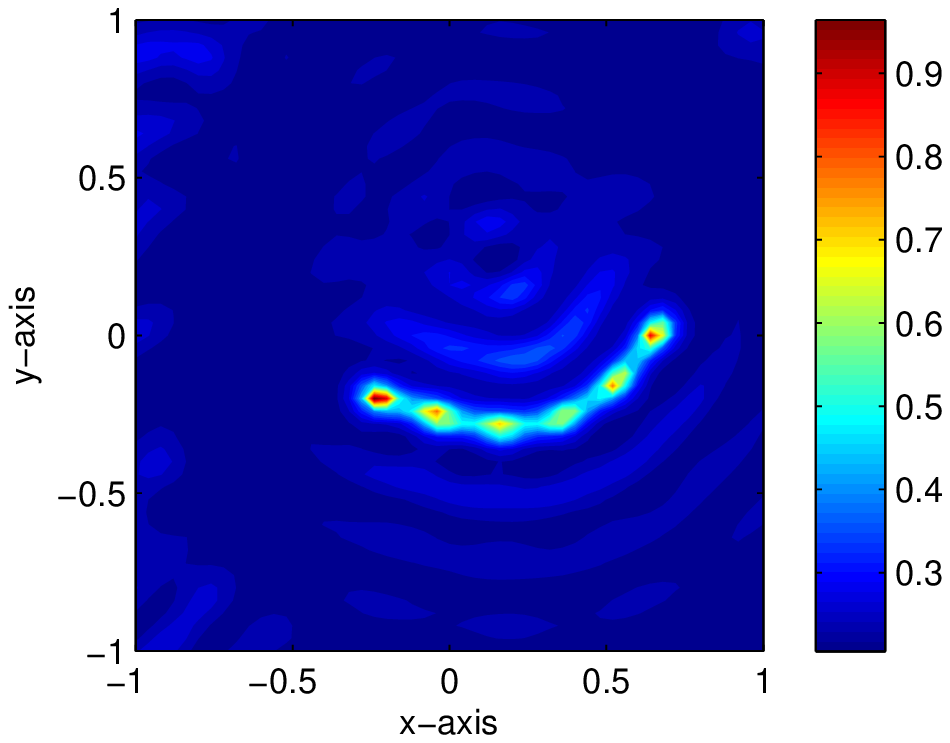}}
  \subfigure[Permeability contrast case]{\label{MapMUSIC2}\includegraphics[width=0.49\textwidth]{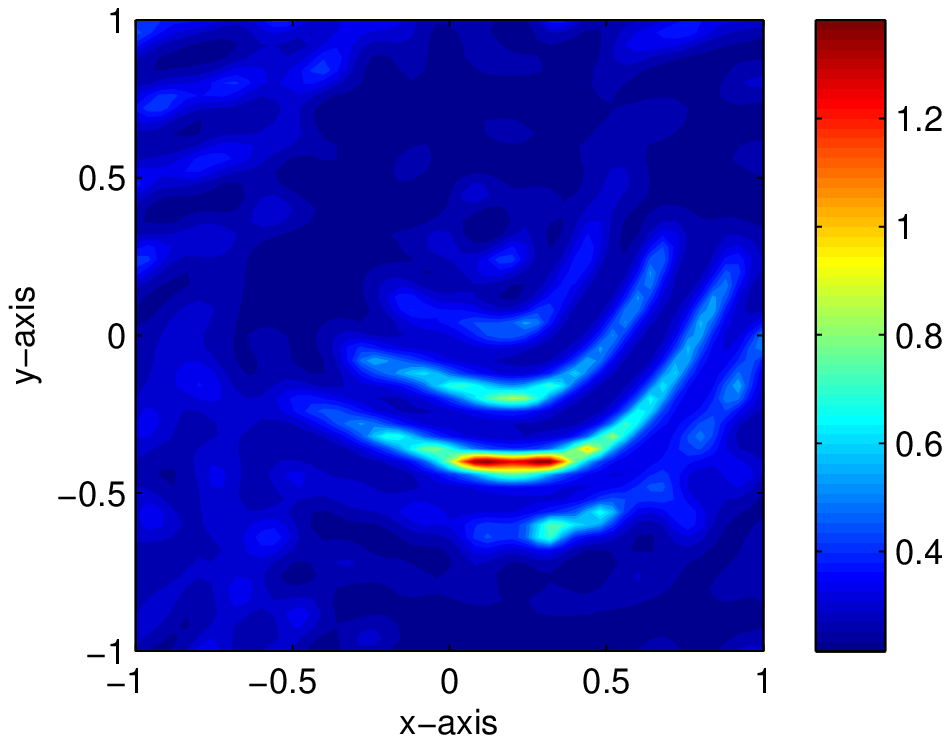}}
  \caption{\label{MapMUSIC}Shape reconstruction of thin electromagnetic inclusion $\Gamma_1$ using MUSIC algorithm.}
  \end{center}
\end{figure}

For a numerical example of a perfectly conducting crack, we selected the following smooth curve in the form of (\ref{PCrack}):
\[\Gamma_2=\left\{\left[x,\frac{1}{2}\cos\frac{x\pi}{2}+\frac{1}{5}\sin\frac{x\pi}{2}-\frac{1}{10}\cos\frac{3x\pi}{2}\right]^T:-1\leq x\leq1\right\}.\]
The data set for the MSR matrix $\mathbb{K}$ in (\ref{MSR}) is collected by solving the forward problems introduced in \cite[Chapter 3]{N} and \cite[Chapter 4]{N} for the sound-soft and sound-hard arcs, respectively. FIG. \ref{MapMUSICPC} shows maps of $\mathbb{E}(\mz)$ for $N=40$ directions and $\lambda=0.4$. By comparing the results in FIG. \ref{MapMUSIC}, we can observe that Theorems \ref{Theorem1} and \ref{Theorem2} hold for the sound-soft and sound-hard arcs, respectively. Additional numerical results can be found in recent works \cite{PL1,PL3}.

\begin{figure}[!ht]
  \begin{center}
  \subfigure[Sound-soft case]{\label{MapMUSIC4}\includegraphics[width=0.49\textwidth]{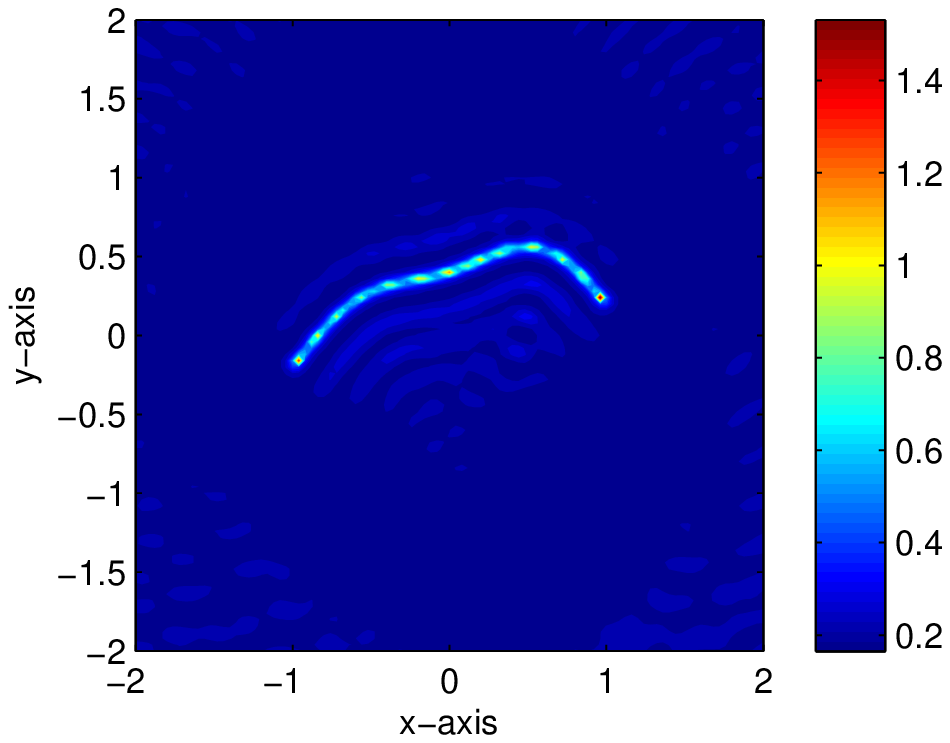}}
  \subfigure[Sound-hard case]{\label{MapMUSIC5}\includegraphics[width=0.49\textwidth]{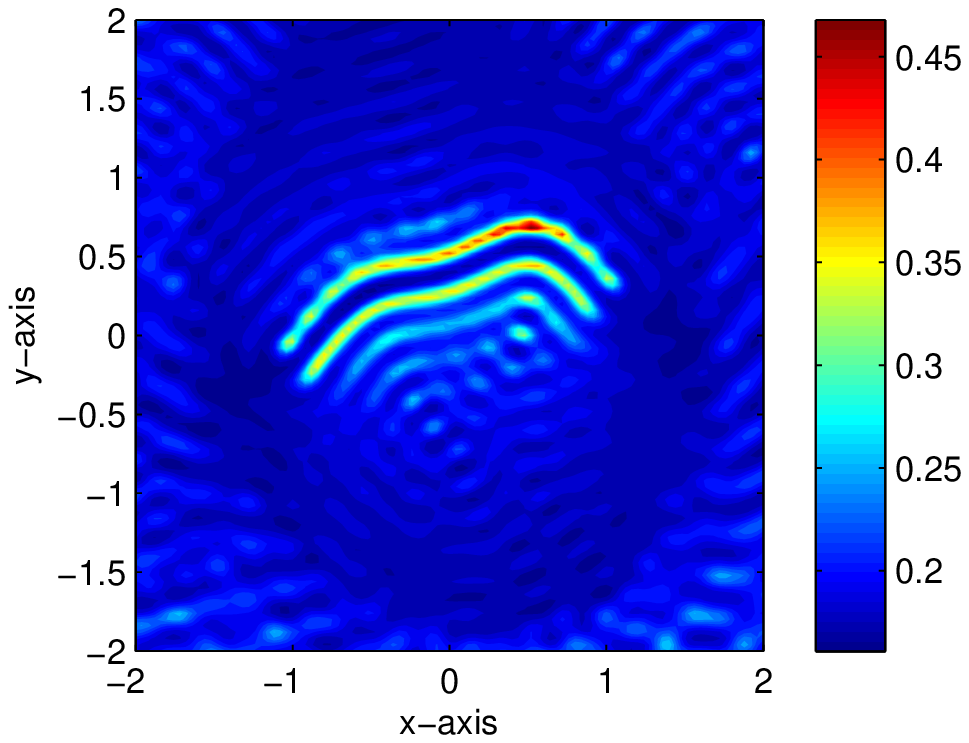}}
  \caption{\label{MapMUSICPC}Shape reconstruction of perfectly conducting crack $\Gamma_2$ using MUSIC algorithm.}
  \end{center}
\end{figure}

\section{Concluding remarks}\label{sec:5}
On the basis of the structure of the left singular vectors of the MSR matrix, we investigated the structure of the MUSIC-type imaging function by establishing a relationship between it and the Bessel function of integer order of the first kind. Using this relationship, we examined certain properties of the MUSIC algorithm.

It is worth emphasizing that the MUSIC algorithm can be applied in limited-view inverse scattering problems. However, its structure has been identified for the sound-soft arc of small length \cite{JKP}. Hence, exploring the structure of MUSIC for the extended, sound-hard arc will be a future work.

Finally, we have been considering the imaging of two-dimensional thin electromagnetic inclusions or perfectly conducting cracks. The analysis could be extended to a three-dimensional problem; refer to \cite{AILP} for related work.


\begin{thebibliography}{99}
\bibitem{AGJKLY}{\sc H. Ammari, P. Garapon, F. Jouve, H. Kang, M. Lim, and S. Yu}, {\em A new optimal control approach for the reconstruction of extended inclusions}, SIAM J. Control. Optim., 51 (2013), pp. 1372--1394.
\bibitem{AGKLS}{\sc H. Ammari, J. Garnier, H. Kang, M. Lim, and K. S{\o}lna}, {\em Multistatic imaging of extended targets}, SIAM J. Imaging Sci., 5 (2012), pp. 564--600.
\bibitem{AGKPS}{\sc H. Ammari, J. Garnier, H. Kang, W.-K. Park, and K. S{\o}lna}, {\em Imaging schemes for perfectly conducting cracks}, SIAM J. Appl. Math, 71 (2011), pp. 68--91.
\bibitem{AIL}{\sc H. Ammari, E. Iakovleva, and D. Lesselier}, {\em A MUSIC algorithm for locating small inclusions buried in a half-space from the scattering amplitude at a fixed frequency}, Multiscale Model. Simul., 3 (2005), pp. 597--628.
\bibitem{AILP}{\sc H. Ammari, E. Iakovleva, D. Lesselier, and G. Perrusson}, {\em MUSIC type electromagnetic imaging of a collection of small three-dimensional inclusions}, SIAM J. Sci. Comput. 29 (2007), pp. 674--709.
\bibitem{AK}{\sc H. Ammari, and H. Kang}, {\em Reconstruction of Small Inhomogeneities from Boundary Measurements}, Lecture Notes in Mathematics, 1846, Springer-Verlag, Berlin, 2004.
\bibitem{AKKLV}{\sc H. Ammari, H. Kang, E. Kim, M. Lim, and K. Louati}, {\em A direct algorithm for ultrasound imaging of internal corrosion}, SIAM J. Numer. Anal., 49 (2011), pp. 1177--1193.
\bibitem{AKLP}{\sc H. Ammari, H. Kang, H. Lee, and W.-K. Park}, {\em Asymptotic imaging of perfectly conducting cracks}, SIAM J. Sci. Comput., 32 (2010) pp. 894--922.
\bibitem{BHL}{\sc G. Bao, S. Hou and P. Li}, {\em Inverse scattering by a continuation method with initial guesses from a direct imaging algorithm}, J. Comput. Phys., 227 (2007), pp. 755--762.
\bibitem{BF}{\sc E. Beretta, and E. Francini}, {\em Asymptotic formulas for perturbations of the electromagnetic fields in the presence of thin imperfections}, Contemp. Math., 333 (2003), pp. 49--63.
\bibitem{B}{\sc M. Burger}, {\em A level set method for inverse problems}, Inverse Problems, 17 (2001), pp. 1327--1355.
\bibitem{CR}{\sc A. Carpio, and M.-L. Rapun}, {\em Solving inhomogeneous inverse problems by topological derivative methods}, Inverse Problems, 24 (2008), 045014.
\bibitem{CZ}{\sc X. Chen, and Y. Zhong}, {\em MUSIC electromagnetic imaging with enhanced resolution for small inclusions}, Inverse Problems, 25 (2009), 015008.
\bibitem{C}{\sc M. Cheney}, {\em The linear sampling method and the MUSIC algorithm}, Inverse Problems 17 (2001), pp. 591--595.
\bibitem{D}{\sc A. J. Deveney}, {\em Super-resolution processing of multi-static data using time-reversal and MUSIC}, available at \url{http://www.ece.neu.edu/faculty/devaney/preprints/paper02n_00.pdf}
\bibitem{DL}{\sc O. Dorn, and D. Lesselier}, {\em Level set methods for inverse scattering}, Inverse Problems, 22 (2006), pp. R67--R131.
\bibitem{DMR}{\sc O. Dorn, E. L. Miller, and C. M Rappaport}, {\em A shape reconstruction method for electromagnetic tomography using adjoint fields and level sets}, 16 (2000), pp. 1119--1156.
\bibitem{G}{\sc R. Griesmaier}, {\em Multi-frequency orthogonality sampling for inverse obstacle scattering problems}, Inverse Problems, 27 (2011), 085005.
\bibitem{GH1}{\sc R. Griesmaier, and M. Hanke}, {\em MUSIC-characterization of small scatterers for normal measurement data}, Inverse Problems, 25 (2009), 075012.
\bibitem{GH2}{\sc R. Griesmaier, and N. Hyv\"{o}nen}, {\em A regularized Newton method for locating thin tubular conductivity inhomogeneities}, Inverse Problems, 27 (2011), 115008.
\bibitem{HM}{\sc H. Haddar, and R. Mdimagh}, {\em Identification of small inclusions from multistatic data using the reciprocity gap concept}, Inverse Problems, 28 (2012), 045011.
\bibitem{HSZ}{\sc S. Hou, K. S{\o}lna, and H. Zhao}, {\em A direct imaging algorithm for extended targets}, Inverse Problems, 22 (2006), pp. 1151--1178.
\bibitem{JKP}{\sc Y.-D. Joh, Y. M. Kwon, and W.-K. Park}, {\em MUSIC-type imaging of perfectly conducting cracks in limited-view inverse scattering problems}, Appl. Math. Comput., in revision.
\bibitem{JP}{\sc Y.-D. Joh, and W.-K. Park}, {\em Structural behavior of the MUSIC-type algorithm for imaging perfectly conducting cracks}, Prog. Electromagn. Res., 138 (2013), pp. 211--226.
\bibitem{K}{\sc R. Kress}, {\em Inverse scattering from an open arc}, Math. Methods Appl. Sci., 18 (2003), pp. 267--293.
\bibitem{M2}{\sc L. M\"onch}, {\em On the inverse acoustic scattering problem by an open arc: the sound-hard case}, Inverse Problems, 13 (1997), pp. 1379--1392.
\bibitem{N}{\sc Z. T. Nazarchuk}, {\em Singular Integral equations in Diffraction Theory}, Karpenko Physicomechanical Institute, Ukrainian Academy of Sciences, 210, Lviv, 1994.
\bibitem{NK}{\sc Z. Nazarchuk, and K. Kobayashi}, {\em Mathematical modelling of electromagnetic scattering from a thin penetrable target}, Prog. Electromagn. Res., 55 (2005), pp. 95--116.
\bibitem{P1}{\sc W.-K. Park}, {\em  Analysis of a multi-frequency electromagnetic imaging functional for thin, crack-like electromagnetic inclusions}, Appl. Numer. Math, 77 (2014), pp. 31--42.
\bibitem{PL1}{\sc W.-K. Park, and D. Lesselier}, {\em Electromagnetic MUSIC-type imaging of perfectly conducting, arc-like cracks at single frequency}, J. Comput. Phys., 228 (2009), pp. 8093--8111.
\bibitem{PL3}{\sc W.-K. Park, and D. Lesselier}, {\em MUSIC-type imaging of a thin penetrable inclusion from its far-field multi-static response matrix}, Inverse Problems, 25 (2009), 075002.
\bibitem{PL4}{\sc W.-K. Park, and D. Lesselier}, {\em Reconstruction of thin electromagnetic inclusions by a level set method, Inverse Problems}, 25 (2009), 085010.
\bibitem{SZ}{\sc J. Soko{\l}owski, and A. Zochowski}, {\em On the topological derivative in shape optimization}, SIAM J. Control Optim., 37 (1999), pp. 1251--1272.
\bibitem{SCC}{\sc R. Song, R. Chen, and X. Chen}, {\em Imaging three-dimensional anisotropic scatterers in multi-layered medium by MUSIC method with enhanced resolution},  J. Opt. Soc. Am. A, 29 (2012), pp. 1900--1905.
\bibitem{VXB}{\sc G. Ventura, J. X. Xu, and T. Belytschko}, {\em A vector level set method and new discontinuity approximations for crack growth by EFG}, Int. J. Numer. Meth. Engng, 54 (2002), pp. 923--944.
\bibitem{ZC}{\sc Y. Zhong, and X. Chen}, {\em MUSIC imaging and electromagnetic inverse scattering of multiple-scattering small anisotropic spheres}, IEEE Trans. Antennas Propag., 55 (2007), pp. 3542--3549.
\end{thebibliography}
\end{document}